\documentclass[preprint,12pt]{elsarticle}
\usepackage[T1]{fontenc}
\usepackage{ulem}
\usepackage{amssymb}
\usepackage{amsthm}
\usepackage{dsfont}
\usepackage{mathtools}
\usepackage[active]{srcltx}
\usepackage{multirow}
\usepackage{verbatim}
\usepackage{xcolor}
\usepackage{tikz}
\usepackage{pgfplots}
\usepackage{subcaption}
\usepackage{longtable}
\usepackage{dirtytalk}
\usepackage{bm}% boldmath
\usepackage{bold-extra}
 \pgfplotsset{width=7cm}

\newcommand{\iid}{\textit{i.i.d.\ }}

\newcommand{\N}{\mathbb{N}}
\newcommand{\R}{\mathbb{R}}
\newcommand{\NN}{\mathcal{N}_M}
\newcommand{\D}{\mathcal{D}}
\renewcommand{\P}{\mathbf{P}}
\newcommand{\eps}{\varepsilon}
\renewcommand{\epsilon}{\varepsilon}
\newcommand{\fl}[1]{\lfloor #1 \rfloor}

\let\ddd\star
\renewcommand{\star}{\mbox{\large${\ddd}$}}

\renewcommand{\theenumi}{\alph{enumi}}
\renewcommand{\labelenumi}{\theenumi)}
\newtheorem{Theorem}{Theorem}
\newproof{Proof}{Proof}

\allowdisplaybreaks

\journal{ Mathematics and Computers in Simulation}

\begin{document}

\begin{frontmatter}



\title{Improving the Quality of Random Number Generators by Applying a Simple
  Ratio Transformation}


\author[CL]{Michael Kolonko\corref{cor1}%\fnref{fn1}
}
\ead{kolonko@math.tu-clausthal.de}

\author[BE]{Zijun Wu\corref{cor2}%\fnref{fn2}
}
\ead{zijunwu@bjut.edu.cn}

\author[CL]{Feng Gu%\fnref{fn3}
}
\ead{quagfeng@gmail.com}

\cortext[cor2]{Principal corresponding author}
\cortext[cor1]{Corresponding author}

% \fntext[fn1]{Professor in stochastic optimization group at Technical University of Clausthal}
% \fntext[fn2]{Ph.D. student supervised by  Prof. Dr. M. Kolonko}
% \fntext[fn3]{Master student supervised by Prof. Dr. M. Kolonko}

\address[CL]{Institut f{\"u}r Angewandte Stochastik und Operations Research, Technical University of Clausthal,  Erzstr.\ 1,
D-38678 Clausthal-Zellerfeld, Germany}

\address[BE]{Beijing Institute for Scientific and Engineering Computing, School of Applied Mathematics and Physics, Beijing University of Technology,
Pingleyuan 100, 100124, Chaoyang, Beijing, China
}

\begin{abstract} % 211 words  
  It is well-known that the quality of random number generators can often be
  improved by combining several generators, e.g.\ by summing or subtracting
  their results.  In this paper we investigate the \emph{ratio} of two random
  number generators as an alternative approach: the smaller of two input random
  numbers is divided by the larger, resulting in a rational number from $
  [0,1]$. %This approach was first suggested 
  %in \cite{morgenstern2007uniform}.

  We investigate theoretical properties of this approach and show that it yields
  a good approximation to the ideal uniform distribution. To evaluate the
  empirical properties we use the well-known test suite \textsc{TestU01}. 
  %of
  %\cite{l2007testu01}. 
  We apply the ratio transformation to %the 
  moderately bad
  generators, %considered in \cite{l2007testu01}, 
  i.e.\ those that failed up to
  40\% of the tests %\Gu{ 
  from the test battery  \textsc{Crush} of %} \de{ from } 
  \textsc{TestU01}. We show that more than half of them turn into very good
  generators that pass all tests of \textsc{Crush} and \textsc{BigCrush} from %\Gu{
  \textsc{TestU01} %}
   % \cite{l2007testu01} 
   when the ratio transformation is applied.  In
  particular, generators based on linear operations seem to benefit from the
  ratio, as this breaks up some of the unwanted regularities in the input
  sequences.

  Thus the additional effort to produce a second random number and to calculate
  the ratio allows to increase the quality of available random number generators.
\end{abstract}

\begin{keyword}


random number generation \sep combination of random generators \sep ratio of uniform %\Gu{ 
random variables %} \de{rvs } 
\sep empirical testing of random generators  
\end{keyword}

\end{frontmatter}


\section{Introduction}

Stochastic simulation is an important tool to study the behavior of complex
stochastic systems that cannot be %\Wuu{
mathematically analyzed,
%}{fully analyzed mathematically}, 
as it is often
the case e.g.\ in network models for traffic, communication or production. The
basis of these simulation tools is a (pseudo-) random input provided by a  random
number generator.

A random number generator (RNG for short) is an algorithm  that produces sequences
of numbers that, viewed as an observation of a random experiment, can be modeled
mathematically by a sequence of independent, identically distributed random
variables (\iid rvs). As basis for most simulations, these rvs should have the
uniform distribution $ U(0,1)$ on the interval $ [0,1]$.

Of course, a deterministic algorithm can only approximate this mathematical
model. Its deviance from the model is used to measure the quality of the
generator. A lot of investigations have been made into that direction for
different types of generators proposed over time, see e.g.\ \cite{knuth1998art}
or \cite{l2007testu01} for a survey. In particular, a large number of empirical
tests have been developed to assess the quality of RNGs.  
% To be able to pass at
% least most of these tests, generators often have to use complex and
% time-consuming operations, see e.g.\ \cite{matsumoto1998mersenne},
% \cite{l1999tables} or \cite{eichenauer1991inversive} (though there are
% exceptions as e.g.\ \Kol{MRG32k3a} given in \cite{l1999good}).

In this paper we are going to show that one can transform many simple,
moderately good generators into statistically excellent ones using the ratio of
their output. We show this with the test suite \textsc{TestU01} from
\cite{IP-TestU01}, which has become a standard for RNG testing.

A general framework for RNGs producing numbers in the interval $ [0,1]$ is
described in \cite{l1988efficient}. It consists of a finite set $ S$ of internal
\emph{states}, a function $ f:S\to S$ describing the \emph{recursion}, a
\emph{seed state} $ s_0$ and an \emph{evaluation function} $ g: S\to
[0,1]$. Starting with the seed state $ s_0$, a sequence of states $ (s_i)_{i\ge
  0}$ is constructed using the recursion $ s_{i+1}:=f(s_i), i\ge 0$. The random
numbers returned are $ u_i:=g(s_i), i\ge 0$.

Well-known examples for RNGs are the linear congruential generators (LCG) of
order $ 1$. Here,  $ f(s):=(as+c)\mod M$ \ for some constants $ a,c,M\in \N$ and $
S:=\{0,1 ,\ldots, M-1\}$. The integers $ s$ that serve as internal states are
turned into output values from $ [0,1]$ by the evaluation function $ g(s):=s/M$.  

In many popular generators the internal state consists of one or more integers
from a bounded set $\NN:=\{0,1 ,\ldots, M-1\}$. The evaluation step then selects
one of these integers and divides it by $ M$ as in the example above.  This is
the case e.g.\ in LCGs of higher order, in many so-called Lagged Fibonacci
generators and more complex combinations of RNGs as given in
\cite{knuth1998art}, \cite{l1999good}, \cite{l2001software} or
\cite{matsumoto1998mersenne}. This way of producing random numbers from $ [0,1]$
using a division by a constant $ M$ will be referred to in the sequel as
the \emph{direct approach}. 

%In our ratio approach defined below, we replace
%this step by the division of two integers produced by the generator.

In this paper, we assume that we are given a RNG that internally produces
integers $ x_i\in \NN$ as above, but the final evaluation step 
%then the final step $ x_i/M$ is replaced by
is a more complex transformation that uses two consecutive values $
x_{2i},x_{2i+1}\in \NN$ and returns
\begin{equation}\label{eq:Def1}
u_i :=  \frac{ \min \{x_{2i},x_{2i+1} \} }{ \max \{ x_{2i},x_{2i+1} \} }
\end{equation}
where the cases $x_{2i}\cdot x_{2i+1}=0$ and $x_{2i}= x_{2i+1}$ will be
considered in detail %\Wuu{
in Subsection \ref{ssec:cdf}.
%}{below}. 
We shall call \eqref{eq:Def1} the \textit{ratio
  transformation} and $ x_1,x_2 ,\ldots $ its \emph{input sequence} from the
\emph{base RNG}. To our knowledge, this type of ratio transformation was first
used in \cite{morgenstern2007uniform}.

To make the comparison between
the ratio transformation and the direct approach fairer, we will consider an
extension of this   approach that also uses two random numbers and returns
\begin{equation*}%\label{eq:DefStd2}
w_i:=  \frac{x_{2i}}{M}+\frac{x_{2i+1}}{M^2}.
\end{equation*}
We call this the \emph{direct-2 approach}.  Note that our ratio
transformation is computationally slightly more complex than direct-2.

% But as is shown below, many popular RNGs have a much improved
%   empirical performance under the ratio transformation compared to both standard
%   approaches so that this additional effort seems to be justified.

We start the paper with a theoretical analysis of the ratio transformation. We
show that the cumulative distribution function (cdf) of the ratio transformation
of two independent discrete random variables $ X_1,X_2$, that are both uniformly
distributed over $ \NN$, approximates the continuous uniform distribution on $
[0,1]$.  This approximation, though not as close as by the direct
  approaches, yields values that are much less regularly distributed than those
  from the direct approaches. We believe that this is the reason for the 
  % is not as close as that of direct approaches standard-2 or even the simple
  % standard approach $ X_1/M$ but the values that the ratio may attain are much
  % less regular than those of the standard approaches leading to a}
 superior empirical behaviour.

The empirical quality of the ratio transformation is tested using the standard
test %\de{suites} \Gu{
batteries %} 
\textsc{Crush} and \textsc{BigCrush} from
\cite{l2007testu01}. Typically, simple RNGs fail many of these tests with the
direct approach. However, if the ratio transformation is applied, the number
of tests failed is reduced considerably. In particular, we tested the RNGs from
\cite{l2007testu01} that failed up to 40\% of the 144 tests of
\textsc{Crush}. When the ratio transformation is applied to their output, about half
of them passed \emph{all} tests of \textsc{Crush}, whereas an application of the
direct-2 approach could hardly improve their performance. This shows that the
ratio transformation is able to turn simple, moderately good RNGs into excellent
ones.  First observations in that direction were reported in
\cite{morgenstern2007uniform}.
%\Kol{Dropped}

The paper is organized as follows: in Section \ref{sec:Theoretical}, we
investigate theoretical properties of \eqref{eq:Def1} under the assumption that
the inputs are from ideal RNGs. In Section \ref{sec:Exp} we report on empirical
tests with the \textsc{Crush}/\textsc{BigCrush} test %\de{suite} \Gu{
battery %} 
from
\cite{l2007testu01} with different types of base RNGs, each time comparing the
direct approaches with the ratio transformation. Finally, we give a short
conclusion.

\section{The Mathematical Model}
\label{sec:Theoretical}
\subsection{The Cumulative Distribution Function of the Ratio Transformation}
\label{ssec:cdf}

  In this Section we investigate the mathematical model of the ratio
  \eqref{eq:Def1} where we replace the input sequence by random variables with a
  uniform distribution. We prove that the ratio transformation preserves
  uniformity at least approximately.

  For completeness, we first give a simple theorem %\Wuu{
  which states
  %s}{that shows} 
  that the ratio of
  two \emph{continuous} $ U(0,1)$-distributed independent random variables is
  again $ U(0,1)$-distributed. %\Wuu{
  Then,
  %}{Below}
   we will study the more complicated
  situation where the input is from a \emph{discrete} uniform distribution.

  Note that for a continuous random variable $ U$ with distribution $ U(0,1)$,
  its cdf is $ F_U(t)=\P(U\le t)=t$ and its density
  $ f_U(t)=1$ for $t\in [0,1]$.
%%%%%%%%%%%%%%%%%%%%%%%%%%%%%%%%%%%%%%%%%%%%%%%%%%%%%%%%%%%%%%%%%%%%%%%%%%%%%%%% 
\begin{Theorem}
  \label{theo:RatioCont}
Let $ U_1,U_2$ be two independent and identically  distributed  random
variables with distribution $ U(0,1)$ and define the ratio
\begin{equation}\label{eq:DefRatioCont}
Z:=  \frac{\min\{U_1,U_2\}}{\max\{U_1,U_2\}}
\end{equation}
where we put $ \frac{0}{0}=0$. Then $ Z$ has distribution $ U(0,1)$, too.
\end{Theorem}
%%%%%%%%%%%%%%%%%%%%%%%%%%%%%%%%%%%%%%%%%%%%%%%%%%%%%%%%%%%%%%%%%%%%%%%%%%%%%%%% 

\begin{proof}
  As   $ \P(U_1\cdot U_2=0)=\P(U_1=U_2)=0$ we may exclude these two cases and obtain
  \begin{align*}
    \P(Z\le t)&= \P(\frac{\min\{U_1,U_2\}}{\max\{U_1,U_2\}}\le t) = 2\P(0<U_1\le
    tU_2)\\
    &= 2\int_0^1\P(U_1\le tu)\; du = 2\int_0^1 tu\;du = 2t/2 =t,
  \end{align*}
for all $ t\in [0,1]$. 
\end{proof}

Now we %\Wuu{}{will} 
assume that the input is from two \iid random variables $ X_1,X_2$
with the \emph{discrete} uniform distribution $ U(\NN)$ on $ \NN=\{0,1 ,\ldots,
M-1\}$, i.e.\ 
\begin{equation}\label{eq:iid}
 \P(X_1=k,X_2=l)=\frac{1}{M^2}\quad \text{  for all } k,l\in \NN.
\end{equation}
%\Wuu{
Then, we
%e}{We then} 
have $ \P(X_1\cdot X_2=0)\approx 2/M$ and $\P( X_1=X_2>0)\approx
1/M.$
%\Kol{Maybe we should drop the following sentence: }\Wuu{. We need}{, and we have} to take care of these two cases. 
A straightforward
extension of the ratio \eqref{eq:Def1} would choose $ 0$ as return value in case
$ X_1\cdot X_2=0$ and return value $ 1$ if $ X_1=X_2>0,$ %\Wu{, 
see \cite{morgenstern2007uniform}. %}.  
As it is a common
practice with RNGs to completely avoid $ 0,1$ as return values, we introduce
%\Wuu{
two replacement values $\epsilon_0$ and $\epsilon_1$ with $ 0<\epsilon_0<1-\epsilon_1<1,$ and let them 
%}{replacement values $ 0<\epsilon_0<1-\epsilon_1<1$ that} 
appear with almost equal
small probability ($ \approx 1.5/M$). We therefore define
\begin{equation}\label{eq:DefEps}
\begin{split}
\eps_0&:= \frac{M-1+\fl{M/2}}{2M^2}\\
\eps_1&:=\frac{2M-1-\fl{M/2}}{2M^2}
  \end{split}
\end{equation}
where $ \fl{a}$ is the largest integer less or equal $ a\in \R$.  Note that $
0<\epsilon_0, \epsilon_1<1/(M-1)$ and that for large $ M,$ $
\epsilon_0\approx \epsilon_1\approx 0.75/M$. The motivation for this
particular choice will become clear from Theorem \ref{theo:KSDistance} below.

In the sequel we will use the following ratio transformation of
two inputs $ x_1,x_2\in \NN$
\begin{equation}\label{eq:DefY}
h(x_1,x_2):=
\begin{cases}
   \eps_0 & \text{if } \  x_1=0<x_2 \ \text{ or }\ 0\le x_1=x_2\le \fl{M/2}-1,\\ 
\frac{ \min\{x _{1}, x _{2}\}}{ \max\{x _{1}, x 	_{2}\}} &\text{if } \ x
_{1} \cdot x _{2} > 0 \ \text{ and }  \   x _{1} \neq x _{2}, \\ 
  1-\eps_1 & \text{if } \  x_2=0<x_1 \quad \text{ or }\quad  x_1=x_2 \ge \fl{M/2},
\end{cases}
\end{equation}
i.e.\ we split the cases $ x_1\cdot x_2=0$ and $ x_1=x_2$ more or less evenly
between the two values $ \epsilon_0$ and $1-\epsilon_1$.    
The next Theorem gives the cdf of the discrete random variable $ h(X_1,X_2)$. We
%\Wuu{
show in Theorem \ref{theo:KSDistance} that
%}{will show below, that} 
this is a close approximation of the cdf of $ U(0,1)$.

%%%%%%%%%%%%%%%%%%%%%%%%%%%%%%%%%%%%%%%%%%%%%%%%%%%%%%%%%%%%%%%%%%%%%%%%%
\begin{Theorem}\label{Theo:cdfRatioDiscr}
  Let $ X_1,X_2$ be \iid $ U(\NN)$-distributed random variables and $
  Y=h(X_1,X_2)$. Then the cdf of $ Y$ is given by
\begin{equation}\label{eq:cdfYDiscr}
\P(Y\le t) = \begin{cases}
0 &\text{ if } \quad 0\le t<\eps_0\\
  2\epsilon_0 &\text{ if } \quad \eps_0\le t <\frac{1}{M-1}\\
2\epsilon_0+ \frac{2}{M^2} \cdot\sum_{k=1}^{M-1} \lfloor tk\rfloor& \text{
  if } \quad \frac{1}{M-1}\le t < 1-\eps_1\\
1 & \text{ if }\quad  1-\eps_1\le t \le 1
\end{cases}
\end{equation}
\end{Theorem}
%%%%%%%%%%%%%%%%%%%%%%%%%%%%%%%%%%%%%%%%%%%%%%%%%%%%%%%%%%%%%%%%%%%%%%%

\begin{proof}
  The smallest nonzero value that $\frac{\min\{X_1,X_2\}}{\max\{X_1,X_2\}} $ may attain is $ \frac{1}{M-1}$ and
  similarly, $\frac{M-2}{M-1}=1-\frac{1}{M-1}$ is the largest value smaller than
  $ 1.$ %\Gu{, which is 
  %of $\frac{\min\{X_1,X_2\}}{\max\{X_1,X_2\}}.$ %} . 
  We have $ 0<\epsilon_0< \frac{1}{M-1}$ and $
  \frac{M-2}{M-1}<1-\epsilon_1$. Therefore, using \eqref{eq:iid}
\begin{align}
  \P(Y\le \epsilon_0)&= \P(Y=\epsilon_0)=\P\big( X_1=0<X_2 \ \text{ or
  }\ 0\le X_1=X_2\le \fl{M/2}-1\big)\notag\\
  &= \P\big( X_1=0<X_2\big) + \P\big(0\le X_1=X_2\le \fl{M/2}-1\big)\notag\\
  &= \sum_{k=1}^{M-1}\P(X_1=0, X_2=k)+ \sum_{k=0}^{\fl{M/2}-1}\P(X_1= X_2=k)\notag\\
  &= \frac{M-1+ \fl{M/2}}{M^2} = 2\epsilon_0. \label{eq:ProbEps0}
\end{align} 
Similarly, one may show
\begin{equation}\label{eq:ProbEps1}
\P(Y=1-\epsilon_1)=2\epsilon_1.
\end{equation}

% \begin{align*}
%   \P(Y=1-\epsilon_1)&= \big(X_2=0<X_1  \text{ or } \fl{M/2} \le  X_1=X_2 \big)\\
%   &= \frac{M-1}{M^2} + \sum_{k=\fl{M/2}}^{M-1}\P(X_1=k,X_2=k)\\
%   &= \frac{2M-1-\fl{M/2} }{M^2} =2\epsilon_1.
% \end{align*}
For $ t\in \big[\frac{1}{M-1},\frac{M-2}{M-1}\big]$ we have, using the symmetry
of the joint distribution of $X_1,X_2$,
\begin{align*}
  \P(Y\le t)&= \P(Y=\epsilon_0)+\P( Y\le t, X_{1} \cdot X _{2} > 0, X _{1} \neq
  X _{2})\\
  &= 2\epsilon_0+\P(\frac{\min\{X_1,X_2\}}{\max\{X_1,X_2\}} \le t,\ X_1\cdot
  X_2> 0)\\
%  &= 2\epsilon_0+\P(\min\{X_1,X_2\} \le t\cdot \max\{X_1,X_2\},\  X_1\cdot X_2> 0)\\
  &= 2\epsilon_0+ 2\cdot \P(X_1 \le tX_2,\  X_1> 0)\\
  &= 2\epsilon_0+ 2 \sum_{k=1}^{M-1}\P(X_2=k)\P(X_1\le kt)= 2\epsilon_0+
  \frac{2}{M^2}\sum_{k=1}^{M-1}\fl{tk},
\end{align*}
here $ t<1$ implies $ X_1\not=X_2$, which proves the Theorem.
\end{proof}

%%%%%%%%%%%%%%%%%%%%%%%%%%%%%%%%%%%%%%%%%%%%%%%%%%%%%%%%%%%%%%%%%%%%%%%%%%%%%%%%%

\subsection{The Deviation of $ Y$ from the Uniform Distribution}  
We want to measure the quality of the ratio transformation $ Y$ by the maximal
deviation of its cdf $ F_Y(t):=\P(Y\le t)$ as given in Theorem
\ref{Theo:cdfRatioDiscr} from the cdf of a  $ U(0,1)$-distributed random
variable $ U$ with $ F_U(t)=t, t \in [0,1]$. This difference
\[ 
\Delta_Y:=\sup_{t\in [0,1]}|F_Y(t)-t|
\]   
is also called Kolmogoroff-Smirnov-distance (KS-distance).

\begin{figure}[htb]
  \centering
  \includegraphics[width=6.3cm]{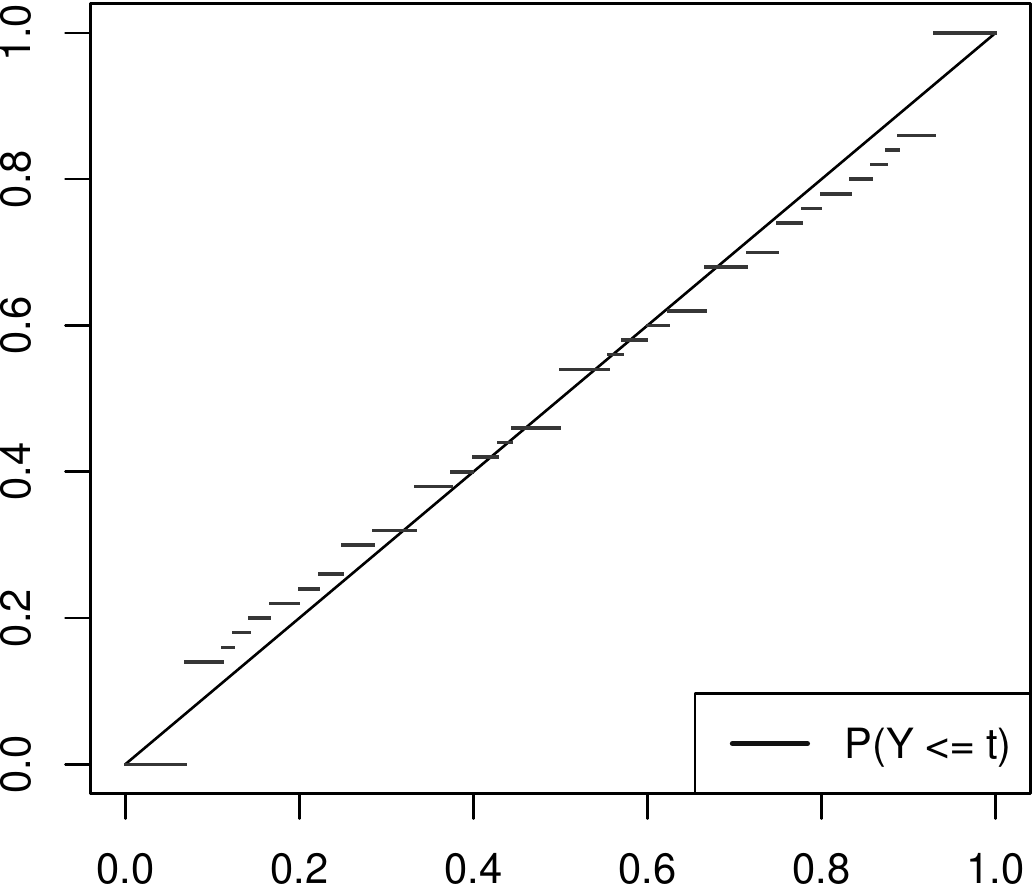}\quad
  \caption{The cdf of the ratio transformation for $ M=10$. The diagonal
    corresponds to the cdf of $ U(0,1)$.}
  \label{fig:CDFStd-Eps}
\end{figure}

Theorem \ref{theo:KSDistance} shows, that the distribution of the ratio
transformation $ Y$ approaches the true $ U(0,1)$ distribution in KS-distance as
$ M\to \infty $. Moreover, the largest deviation of $ F_Y$ from $ F_U$
  takes place at the two artificial extreme points $ \epsilon_0,1-\epsilon_1$.

%%%%%%%%%%%%%%%%%%%%%%%%%%%%%%%%%%%%%%%%%%%%%%%%%%%%%%%%%%%%%%%%%%%%

\begin{Theorem}
  \label{theo:KSDistance}
Let $ Y$ and  $ \Delta_Y$ be defined as above. Then
\[
  \Delta_Y \le \max\{\epsilon_0,\epsilon_1\} \approx \frac{3/4}{M}.
\] 

\end{Theorem}
%%%%%%%%%%%%%%%%%%%%%%%%%%%%%%%%%%%%%%%%%%%%%%%%%%%%%%%%%%%%%%%%%%%%
\begin{proof}
  First note that for the KS-distance between an increasing step function $ F$  with jump
  points $ t_i, i=1 ,\ldots, L$, and the continuous cdf $F_U(t)=t$  on $
  [0,1]$, it is sufficient to check the distances at the jump points, more
  precisely with $ t_0:=0$
  \begin{equation}\label{eq:KSDistance}
\sup_{t\in[0,1]}|F(t)-t|= \max_{i=1 ,\ldots,
  L}\big\{|F(t_{i})-t_{i}|,|F(t_{i-1})-t_i|\big\}.
  \end{equation} 
 Determining $ \Delta_Y$ turns out to be quite involved and we will only
  sketch the most important steps here. First note that $ F_Y$ has jumps at $
  \epsilon_0, 1-\epsilon_1$ and the possible values of the ratio expression $
  k/l\in \D$ where
  \begin{equation}\label{eq:DefD}
\D:=\Big\{ \frac{k}{l}\mid k,l\in \NN, 1 \le  k<l \text{ and } k,l \text{ coprime}\Big\}.  
  \end{equation}
  From \eqref{eq:ProbEps0} and \eqref{eq:ProbEps1} we see that the maximal
  deviation between $ F_Y$ and $ F_U$ outside of $ \D$ is $
  \max\{\epsilon_0,\epsilon_1\}$.  For the proof that $
  \max\{\epsilon_0,\epsilon_1\}$ is also the bound for $ \Delta_Y$, it only
  remains to show that
\begin{equation}\label{eq:DiffonD}
\max\Big\{|F_Y(\epsilon_0)-\frac{1}{M-1}|, \ \sup_{t\in
  [\frac{1}{M-1},\frac{M-2}{M-1}]}|F_Y(t)-t|\Big\} \le \max\{\epsilon_0,\epsilon_1\}.
\end{equation}
Note that here the first term has to be included to check the jump of the cdf
from the left at the minimal value $ \frac{1}{M-1}$ of $ \D\subset
\big[\frac{1}{M-1},\frac{M-2}{M-1}\big]$. The long and technical proof of
\eqref{eq:DiffonD} is sketched in the Appendix.

\end{proof}
%%%%%%%%%%%%%%%%%%%%%%%%%%%%%%%%%%%%%%%%%%%%%%%%%%%%%%%%%%%%%%%%%%%%%%%%%%%%%%
As the KS-distance $ \Delta_Y$ is determined by the probability
of the two additional points $ \epsilon_0,1-\epsilon_1$, $ \Delta_Y$ could be
%\Wu{
further %} 
lowered 
%\Wuu{
by replacing $ \epsilon_0$ with $ \epsilon_0', \epsilon_0''\in (0,
\frac{1}{M-1}),$ 
%}{replacing $ \epsilon_0$ by $ \epsilon_0', \epsilon_0''\in (0,\frac{1}{M-1})$} 
each with half the probability of $ \epsilon{_0}$ and
therefore smaller jumps of $ F_Y$. Similarly, one could replace $
1-\epsilon_1$. Then $ \Delta_Y$ would be dominated by the leftmost term %term on the
%lefthand side 
of inequality \eqref{eq:AppendixA} of Theorem \ref{theo:KSAppend} in the Appendix,
%\eqref{eq:AppendixA}  %\de{\eqref{eq:DiffonD}} \Gu{
%($\epsilon_0$ is not defined in \eqref{eq:DiffonD}), %}, 
which is again bounded by %\Wuu{
$(7/4-2\sqrt{2})/M+2/M^2=0.6715/M+2/M^2$
%}{$(7/4-2\sqrt{2})+2/M^2=0.6715/M+2/M^2$} 
as is shown in the Appendix. %\Wuu{
Since %}{as} 
this is
not a real improvement, we stick to the simpler form of the ratio transformation as
given in definition \eqref{eq:DefY}.

\begin{figure}[htb]
  \centering
   \includegraphics[width=6.3cm]{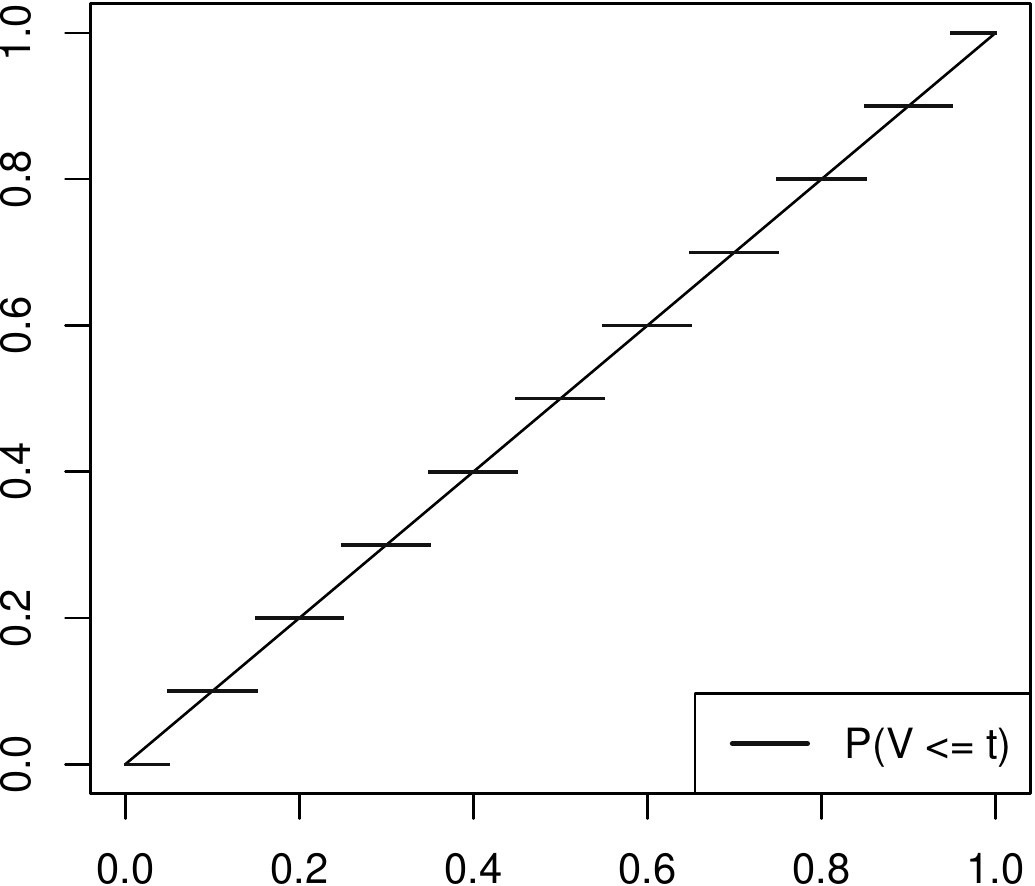}\quad
  \includegraphics[width=6.3cm]{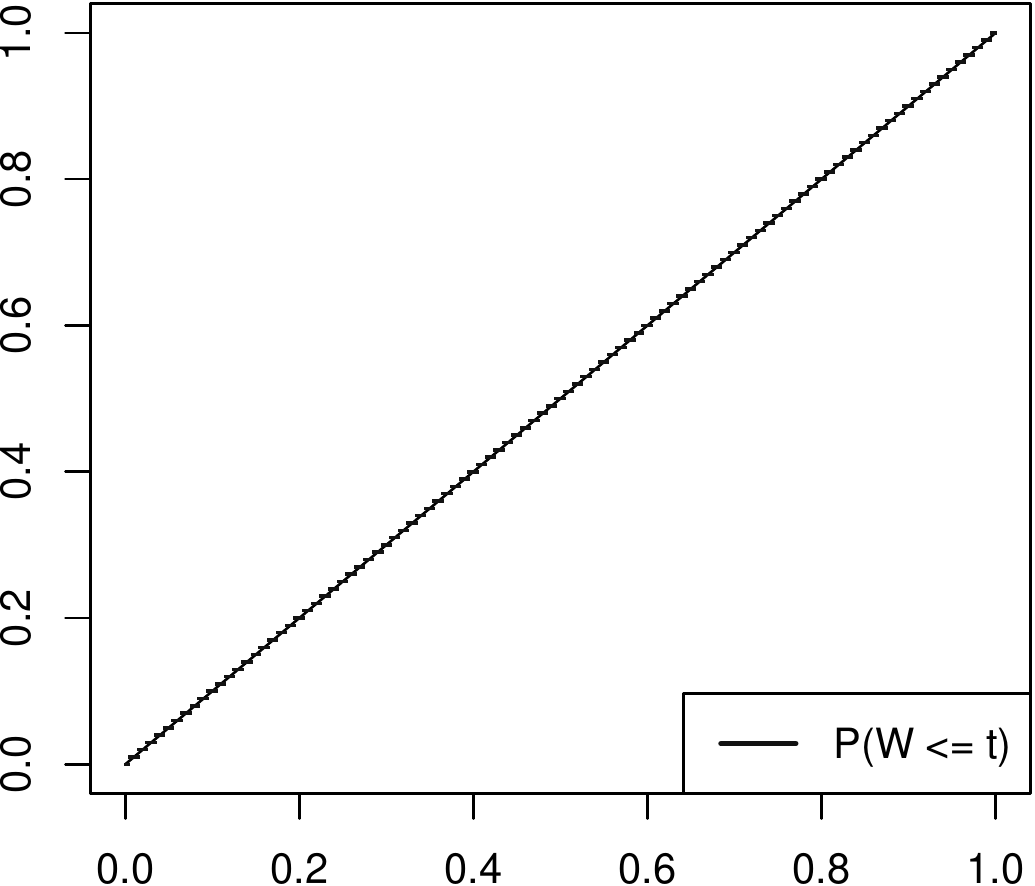}
  \caption{The cdf of the direct approach on the left and the direct-2
    approach on the right for $ M=10$. The diagonal corresponds to the cdf of $
    U(0,1)$.}
  \label{fig:CDFStds}
\end{figure}

Next we want to compare the maximal deviation $ \Delta_Y$ with the deviation
under the two direct approaches. To be more formal, let $ X_1,X_2$ be \iid, $
U(\NN)$-distributed rvs. Then  
\begin{align}
  V&:= \frac{X_1}{M}+\frac{1}{2M} \quad \text{ and} \label{eq:DefStd1}\\
  W&:= \frac{X_{1}}{M}+\frac{X_{2}}{M^2}+\frac{1}{2M^2} \label{eq:DefStd2},  
\end{align}
where we %\Wuu{
use %}{have used} 
a slight shift of the results in both cases again to avoid the value $ 0$.
Then $ V$ corresponds to the simple direct method and $ W$ to direct-2
approach.  A first impression of the cdfs of these variables is given in Figure
\ref{fig:CDFStds}.
We  define
\[
  \Delta_V:=\sup_{t\in [0,1]}|\P(V\le t)-t|,\qquad
  \Delta_W:=\sup_{t\in [0,1]}|\P(W\le t)-t|.
\]
As would be expected from Figure \ref{fig:CDFStds}, at least $ \Delta_W$ is
  much smaller than $ \Delta_Y$. The next Theorem gives simple  bounds on $ \Delta_V,\Delta_W$.
\begin{Theorem}\label{theo:DiffStds}
  With $ V,W$ and $ \Delta_V,\Delta_W$ as above we have
  \[ 
  \Delta_V\ \le \ \frac{1}{2M}\quad \text{ and }\quad \Delta_W\ \le\ \frac{1}{2M^2}  
    \]
\end{Theorem}
\begin{proof} The simple proof is omitted
% We have
%   \[ 
% \P(\frac{X_1}{M}\le t)=\P(X_1\le tM)=\frac{\fl{tM}+1}{M}
%   \]
% and hence 
% \[ 
% F_V(t):=\P(V\le t)=  \frac{\fl{(t-\frac{1}{2M})M}+1}{M}=  \frac{\fl{tM+\frac{1}{2}}}{M}.
% \] 
% $F_V $ is a step function with jump points $t_i:=\frac{i-1/2}{M}, i=1 ,\ldots,
% M-1$. As $ F_V(t_i)-t_i=\frac{1}{2M}$ and $ F_V(t_{i-1})-t_i=-\frac{1}{2M}$, the
% result for $ \Delta_V$ follows from \eqref{eq:KSDistance}, see also Figure \ref{fig:CDFStds}.  A
% similar more involved calculation proves the result for $ \Delta_W$.
% \begin{color}{red}Here the details, to be commented out later:
% \begin{align*}
%   \P(W\le t)&= \P\Big(\frac{X_1}{M}+\frac{X_2}{M^2}\le t - \frac{1}{2M^2}\Big)\\
%  \P\Big(\frac{X_1}{M}+\frac{X_2}{M^2}\le t\Big) &= \sum_{k=0}^{M-1}\frac{1}{M}\P\Big(\frac{X_2}{M^2}\le
%  t-\frac{k}{M}\Big)\\
%  &= \frac{1}{M} \sum_{k=0}^{M-1}\P\Big(X_2\le
%  (t-\frac{k}{M})M^2\Big)\\
%  &= \frac{1}{M} \sum_{k=0}^{M-1} \P(X_2\le \fl{tM^2 -kM})\\
%  &= \frac{1}{M^2} \sum_{k=0}^{M-1}\sum_{l=0}^{\fl{tM^2 -kM}}\P(X_2=l)\\
%  &= \frac{1}{M^2} \sum_{k=0}^{\fl{tM}}\big(\min\{\fl{tM^2 -kM}, M-1\}+1\big)\\
%  &= \frac{1}{M^2}\Big(\sum_{k=0}^{\fl{tM}-1}(M-1+1)+\big(\fl{tM^2
%    -\fl{tM}M}+1\big)\Big)\\
% &\text{ as } \fl{tM^2-kM}\le M-1 \iff k \ge \fl{tM}\\ 
%  &= \frac{1}{M^2}\Big(\fl{tM}M + \fl{tM^2}- \fl{tM}M +1\Big)\\
%  &= \frac{\fl{tM^2}+1}{M^2}\\
%  \P(W\le t)&= \frac{\fl{(t-\frac{1}{2M^2})M^2}+1}{M^2}=  \frac{\fl{tM^2 +\frac{1}{2}}}{M^2}
% \end{align*}
% \end{color}
\end{proof}
 
Theorem \ref{theo:KSDistance} and \ref{theo:DiffStds} show that, with respect to
the KS-distance, both direct
approaches approximate
%\Kol{the }\Wuu{direct-2 approach approximates}{both direct
%  approaches approximate} \Kol{but both approaches are better, the simple direct
%  has 0.5/M compared to 0.6/M for Ratio??} 
the $ U(0,1)$ closer than
the ratio transformation as is also obvious from Figure \ref{fig:CDFStd-Eps}
and \ref{fig:CDFStds}.  Nevertheless, these figure also give a clue why the
ratio transformation performs better in practice: its structure is far more
irregular than that of the direct approaches. 

% \Koll{dropped: it is a restriction of direct-2, which in this context we have introduced}{Note that for large $ M$, the largest value that the direct-2 approach produces,
% $1-\frac{1}{2M^2} $, may be returned as $ 1$ due to the limited precision of the
% data types for floating point numbers. For the ratio transformation and the
% simple direct approach, the maximal values are much farther away from $ 1$.}

\subsection{The Set of Possible Values in $ [0,1]$ } 

The next Theorem gives the exact number of values in $ \D$ as defined in
\eqref{eq:DefD}. Here we use the \emph{Euler totient function} $ \phi(k)$,
  that gives the number of integers $m, 1\le m\le k$, that are relatively prime
  to $ k$.
\begin{Theorem}\label{theo:novalues}
  Assume that $ M > 2,$ %\Gu{(this is not hold for $M = 2$, since $\epsilon _{0} = \epsilon _{1}$ and $M-1<2$ if $M = 2$)},
   then the number of values that the ratio
    transformation $ Y$ may attain is
  \[ 
 N(M):=  \sum_{k=2}^{M-1} \phi(k)+2. 
  \]  
\end{Theorem}
\begin{proof}
We put $ \D_l:=\big\{\frac{k}{m}\mid 0<k<m<l\big\}$, then $
\D=\D_M$. We prove 
\begin{equation}\label{eq:NoT}
|\D_l|=\sum_{k=2}^{l-1}\phi(k)
\end{equation}
for $ l>  2$, then the assertion follows as $ Y$ takes on values in $
  \D\cup\{\epsilon_0,1-\epsilon_1\}$ only. We have  $
\D_3=\{1/2\}$, hence the assertion follows for $ l=3$ as $ \phi(2)=1$.  Now
assume that \eqref{eq:NoT} holds for some $ l\ge 3$. Then all values from $
\D_l$ must be contained in $ \D_{l+1}$. The additional values in $ \D_{l+1}$
must all have the form $k/l$ where $ 1\le k <l$. In order that these
have not yet appeared in $ \D_l$, $ k$ and $ l$ must be relatively prime, hence
there are exactly $ \phi(l)$ additional values and
\[ 
|\D_{l+1}|=|\D_l|+\phi(l)=\sum_{k=2}^{l-1}\phi(k)+\phi(l)=\sum_{k=2}^{l}\phi(k).
\]       
\end{proof}
According to \cite{mitrinovic1996handbook}, for large $ M$,  $ N(M)$ can be approximated by 
\[ 
\frac{3(M-1)^2}{\pi^2}+ O\Big((M-1)\log(M-1)^{2/3}\big(\log\log(M-1)\big)^{4/3}\Big)
\] 
which is much more than the $ M$ different values possible under the direct
approach though it is less than the $ M^2$ different values from the direct-2
approach.

\subsection{Further Theoretical Properties}
\label{ssec:furtherprop}

Nevertheless, the number of different values actually produced by a RNG is
limited by the length of its period. The direct approach obviously keeps the
period length $ T$ of its base RNG.  The period length of pairs $
(x_{2i},x_{2i+1})$ is $ T_1:=T/2$ if $ T$ is even and $ T_1:=2T$ if $ T$ is
odd. So the period length $ T_0$ of the ratio transformation should be a divisor
  of $ T_1$.  Presently, we cannot give any further results on $ T_0$. 

The ratio transformation may also be applied to base RNGs that produce numbers
from $ [0,1]$ (instead of $ \NN$) as is supported by Theorem
\ref{theo:RatioCont}. In case the base RNG uses the direct method, the
denominator $ M$ of the two consecutive random numbers cancels out in the ratio
transformation and Theorem \ref{Theo:cdfRatioDiscr} applies.

 Note that we could save the additional call of the RNG for
the ratio transformation, if we would use overlapping pairs $ (x_i,x_{i+1})$ in
\eqref{eq:Def1}, such that for each step we need to invoke the base RNG only
once. As could be expected, this leads to strongly correlated results and is
therefore not investigated further in this paper.  
If, however, the input values
to the ratio come from \emph{two different} RNGs, one may look at
\eqref{eq:Def1} as a non-linear combination of RNGs that showed an excellent
quality in first empirical tests.

% Such combinations of
% generators are known to improve the quality of the result, but this is often
% depending heavily on the parameters of the two generators (see e.g.\
% \cite{l1988efficient}). With the ratio applied to two separate generators, the
% empirical results seem to be generally better than from single generators as is
% shown in Section \ref{ssec:Crush2Generators}.

The theoretical investigation presented in this Section show that the
ratio transformation might be a candidate to produce good random numbers, but it
is not clear why this %\Wuu{
could %}{should} 
be better than the direct approaches. This will
become evident from the empirical results of the next Section.

\section{Experimental study}
\label{sec:Exp}
\subsection{The Test Set-up}\label{ssec:setup}
The theoretical results of the Section \ref{sec:Theoretical} assumed that the
input sequence was from an ideal generator. We shall now investigate the
empirical behavior of the ratio transformation when the input sequence is from
RNGs used in practice, in particular from fast and simple generators that may not be very
good on their own. We shall show that for many classical generators the
sequences from the ratio transformation are much better than the
sequences obtained from the same RNG with the %\Wuu{
two direct approaches. %}{direct approach.}

We %\Wuu{
use
%}{have used} 
the test batteries \textsc{Crush} and \textsc{BigCrush} from the test suite \textsc{TestU01} of L'Ecuyer and Simard (version 1.2.3 from
\cite{IP-TestU01}), described in \cite{l2007testu01}. In %\Wuu{
that %}{this} 
paper, 92
widely used or well known RNGs from different publications or software
packages are described and tested. The $ H_0$ hypothesis in these tests is ``the
random numbers $ u_1,u_2,\ldots$ are observations from $ U(0,1)$'' and the
software reports the right $ p$-values for each of the tests.  A test is failed
if the $ p$-value lies outside the interval $ [10^{-10}, 1-10^{-10}]$ and
results are considered `suspicious' if the $ p$-value lies in
$[10^{-10},10^{-4}] \cup [ 1-10^{-4}, 1 - 10^{-10}]$ as specified in
\cite{l2007testu01}. There are 144 tests applied in one run of \textsc{Crush},
some of which are based on identical statistical procedures but
with different parameters.

We looked at those 57 RNGs used in \textsc{TestU01} that failed at least one but
not more than 40\% of the 144 tests used in \textsc{Crush}. For each of these
generators we list %\Wuu{
in Table \ref{tb:res} %}{below} 
the original results from
\textsc{TestU01}, the results from an application of our ratio transformation
and the results from an application of direct-2 as in \eqref{eq:DefStd2} under
\textsc{Crush}. If there were no failures or suspicious results in
\textsc{Crush}, %\Wu{for our ratio transformation}\Kol{??? this is also applied to
%  Dir2, maybe better drop this restriction to Ratio} 
we 
%\Wuu{applied further the test suite}{also applied} 
further  applied the \textsc{BigCrush}.
%\
%Wuu{.}{, the results are given in Table
%  \ref{tb:res}.}

We reuse here the names of the generators from
\cite{l2007testu01}, e.g.\ `LCG$(M,a,c)$' denotes a
linear congruential generator with modulus $ M$, multiplicator $ a$ and additive
constant $ c$, for the exact definition of the other generators we refer to
\cite{l2007testu01}.

\newcommand{\xh}{\rule{0pt}{2.5ex}}
\begin{center}
{\scriptsize
\begin{longtable}{|c|l|l|l|l|l|l|l|l|}
\hline \multicolumn{2}{|c|}{} & \multicolumn{3}{c|}{\textbf{\textsc{Crush}}}  &
\multicolumn{2}{c|}{\textbf{\textsc{BigCrush}}} & \multicolumn{2}{c|}{\xh time for $ 10^8$ } \\ 
\multicolumn{2}{|c|}{\textbf{Generator}} & \multicolumn{1}{c|}{\textbf{Dir}}  & \multicolumn{1}{c|}{\textbf{Dir2}}  & \multicolumn{1}{c|}{\textbf{ Rat }}
 & \multicolumn{1}{c|}{\textbf{Dir2}} 
& \multicolumn{1}{c|}{\textbf{ Rat}}& \multicolumn{1}{c|}{\textbf{Dir}}  
& \multicolumn{1}{c|}{\textbf{ Rat }}
 \\ 
\endfirsthead
\multicolumn{9}{c}%
{\xh{\bfseries \tablename\ \thetable{} -- continued from previous page}} \\
\hline \multicolumn{2}{|c|}{} &
\multicolumn{3}{c|}{\textbf{\textsc{Crush}}} &
\multicolumn{2}{c|}{\textbf{\textsc{BigCrush}}} & \multicolumn{2}{c|}{\xh time for $ 10^8$ }\\ 
 \multicolumn{2}{|c|}{\textbf{Generator}} &
\multicolumn{1}{c|}{\textbf{Dir}} &
\multicolumn{1}{c|}{\textbf{Dir2}} &
\multicolumn{1}{c|}{\textbf{Rat}} &
\multicolumn{1}{c|}{\textbf{Dir2}} &
\multicolumn{1}{c|}{\textbf{Rat}} &
\multicolumn{1}{c|}{\textbf{Dir}} &
\multicolumn{1}{c|}{\textbf{Rat}}\\ \hline 
\endhead

\hline \multicolumn{9}{|r|}{{Continued on next page}} \\ \hline
\endfoot

%\hline \hline
\endlastfoot
\hline
\xh$\star$& LCG($2^{46},5^{13},0$) & 38(2) & 38(2) & 0 & & 0 &0.35 & 1.62  \\
$\star$& LCG($2^{48},25214903917,11$) & 21(1) &21(5) & 0 & & 0 & 0.35 & 1.58 \\
$\star$& Java.util.Random & 9(3) & 11(2)& 0 & & 0 & 0.52 & 1.88 \\
$\star$& LCG($2^{48},5^{19},0$) & 21(2) &  27(5) & 0 & & 0 & 0.35 & 1.59  \\
$\star$& LCG($2^{48},33952834046453,0$) & 24(5) & 29(5) & 0 & & 0 & 0.38 & 1.54  \\
$\star$& LCG($2^{48},44485709377909,0$) & 24(5) & 33(7) & 0 & & 0 & 0.36 & 1.58 \\
$\star$& LCG($2^{59},13^{13},0$) & 10(1) & 12(2)& 0 & & 0  & 0.36 & 1.75  \\
$\star$& LCG($2^{63},5^{19},1$) & 5 & 6 & 0 & & 0 & 0.37 & 1.62  \\
$\circ$&LCG($2^{63},9219741426499971445,1$) & 5(1) & 7(2) & (1) & & &&\\  \hline
%-----------------------------------------------------------------
\xh&LCG($2^{31}-1,16807,0$) & 42(9) & 40(5)& 12(5) & & &&\\
%LCG($2^{31}-1,2 ^{15} - 2 ^{10},0$) & 59(7) & 14(5) & \\
&LCG($2^{31}-1, 397204094,0$) & 38(4) &43(3) & 15(1) & & &&\\
&LCG($2^{31}-1, 742938285,0$) & 42(5) &40(2) &  12(6) & & &&\\
&LCG($2^{31}-1, 950706376,0$) & 42(4) & 43(1) &  17(1) & & &&\\
$\star$&LCG($10^{12}-11, 427419669081,0$) & 22(2) &19(5) & 0 & & 0 & 1.2 & 3.2\\
$\star$&LCG($2^{61}-1, 2 ^{30} - 2 ^{19},0$) & 1(4) &7(1) & 0 & & 0& 2.07 & 4.94\\ \hline
\xh$\star$&Wichmann-Hill & 12(3) & 9(3) & 0 & & 0  & 4.16 & 9.51\\
$\star$&CombLec88 & 1 & 2 &0 & &0  & 0.73 & 2.31\\
$\star$&Knuth(38) & 1(1) &1(3) & 0 & &  0 & 0.86 & 3.27% there was a nan if the original ratio is applied
\\  \hline\xh
$\star$&DengLin($2 ^{31}-1,2,46338$) & 11(1) & 1(2) & 0 & & 0 & 1.47 & 4.46\\
$\star$&DengLin($2 ^{31}-1,4,22093$) & 2 & 0& 0 & 0& 0  & 1.47 & 4.47\\  \hline
%-----------------------------------------------------------------
\xh&LFib($2 ^{31}, 55, 24, + $) & 9 &1 & 1(3) & & &&\\
&LFib($2 ^{31}, 55, 24, - $) &11 &1 &3(1) & & &&\\
&ran3 & 11(1) &1(1) & 2(1) & & &&\\
\xh$\circ$&LFib($2 ^{48}, 607, 273, + $) &2 &1 & 0 & & 1&&\\
&Unix-random-64 & 57(6) & 51(8) & 35(2) & & &&\\
&Unix-random-128 & 13 &1 & 6(1) & & &&\\
%?
$\downarrow$&Unix-random-256 & 8 & 0& 15(1) & 0& &&\\  \hline
%-----------------------------------------------------------------
\xh$\star$&Knuth-ran\_array2 & 3 &3 & 0 & & 0 & 1.13 & 3.5\\
&SWB($2 ^{24} , 10, 24$) & 30 & 9(3) & 16(2) & & &&\\
$\circ$&SWB($2 ^{24} , 10, 24$)[$24,48$] & 6(1) &0 & 0 & 7(2)& 1 &&\\
&SWB($2 ^{32}-5 , 22, 43$) & 8 &(2) & 4(2) & & &&\\ 
%?
$\downarrow$&SWB($2 ^{31} , 8, 48$) & 8(2) &14(3) & 9(1) & & &&\\ 
$\downarrow$&Mathematica-SWB & 15(3) & 28(2)  & 16 &&&&\\
$\circ$&SWB($2 ^{32} , 222, 237$) & 2 &0 & (1) & 0&  &&\\   \hline
%-----------------------------------------------------------------
\xh&GFSR($250, 103$) & 8 &77(3) & 2 & & &&\\
$\circ$&GFSR($521, 32$) & 7 &77(2) & (1) & & &&\\
$\circ$&GFSR($607, 273$) & 8 & 77(1) & (1) & & &&\\
$\star$&Ziff98 & 6 &74(4) & 0 & & 0 & 0.38 & 1.83\\
&T800 & 25(4) &14(4) & 5(5) & & &&\\
$\star$&TT800 & 12(4) &8(1) & 0 & & 0& 0.53 & 2.13\\
$\star$&MT19937 & 2 &2 & 0 & & 0 & 0.9 & 3.12\\
$\star$&WELL1024a & 4 & 4& 0 & & 0 & 0.63 & 2.61\\
$\star$&WELL19937a & 2(1) &2 & 0 & & 0& 0.52 & 2.65\\
$\star$&LFSR113 & 6 & 6& 0 & & 0 & 0.56 & 2.16\\
$\star$&LFSR258 & 6 & 6 & 0 & & 0 & 0.68 & 2.65\\
%Marsa-xor32 ($13, 17, 5$) & 59(10) & 6(3) & \\
$\star$&Marsa-xor64 ($13, 7, 17$) & 8(1) &8 &  0 & & 0 & 0.41 & 1.72\\ \hline
%-----------------------------------------------------------------
%?
\xh$\downarrow$&Matlab-rand & 5 &1 & 6(2) & & &&\\
$\star$&Matlab-LCG-Xor & 3 &5 & 0 & & 0  & 0.38 & 1.75\\
$\circ$&SuperDuper-73 & 25(3) &22(5) & 0 & & (1)&&\\
&R-MultiCarry & 40 (4) & 33(4)& 6 & & &&\\
$\star$&KISS93 & 1 & 1& 0 & & 0& 0.42 & 1.82 \\  \hline
%-----------------------------------------------------------------
\xh&ICG($2 ^{31}-1, 1, 1$) & 6 &5(2) & 4 & & &&\\
$\circ$&ICG($2 ^{31}-1, 22211,11926380)$ & 5 &5 & 0 & & 5(6)&& \\
&EICG($2 ^{31}-1, 1, 1$) & 6 & 6(1) & 4 & & &&\\
&EICG($2 ^{31}-1, 1288490188, 1$) & 6 & 5(2) & 4 & & &&\\  \hline
\xh&SNWeyl & 56(12) &65(11) & 19(7) & & &&\\
$\circ$&Coveyou-64 & 1 &1 & (1) & & &&\\\hline
\caption[Results]{No.\ of tests failed in \textsc{Crush} and
  \textsc{BigCrush}. The results are for those 57 RNGs from \cite{l2007testu01},
  which failed at least one and maximally 40\% of the \textsc{Crush} tests
  in their original form. Generators that are much  improved  by the ratio
  transformation are marked with a $ \star$. }
 \label{tb:res} 
\end{longtable}
}
\end{center}

In Table \ref{tb:res}, columns 3-7 show the number of failed or suspicious tests under
different transformations and in the two test batteries.  The number of suspicious
result is given in `()'. `Dir' refers to the original RNG, `Dir2' to the
direct-2 approach and `Rat' to the results from ratio transformation. The first column marks
the success of the ratio transformation: 26  of the 57 RNGs  (= 45.6\%)
 became excellent after an application of the ratio transformation, they
passed all tests in \textsc{Crush} and \textsc{BigCrush} and are marked with a
$\star$. Another 9 RNGs marked with $ \circ$ have a few suspicious results in
\textsc{Crush} or pass all tests in \textsc{Crush} and fail a few in
\textsc{BigCrush}. Those without a mark are RNGs that are improved by the ratio
transformation but still fail with some tests in \textsc{Crush}. For four RNGs only,
results were degraded by an application of the ratio transformation, these are
marked with a $ \downarrow$. 

Note that even the well-known `Mersenne-Twister' MT19937 and its derivatives
WELL1024a and WELL19937a could benefit from the ratio transformation. 
%\Kol{As the code
%of the \textsc{Mathematica} generator was not available, the tests had to make
%external calls to \textsc{Mathematica}, resulting in extreme run times. We therefore stopped the
%tests as soon as it became clear that they were not of type $ \star$.}

The last two columns in Table \ref{tb:res} give the runtime in seconds needed to
produce $ 10^8$ random numbers of the original generator (column `Dir') and of
its ratio transformation (column `Rat') based on the implementations in
\textsc{TestU01}. The ratio transformation is 3 to 4 times slower than the
original RNG.  Note that this can be improved with better implementations
  and, in particular, if the base RNG is integrated into the ratio transform
instead of being called as an external program. Then e.g.\ the ratio
transformation for LCG($2^{63},5^{19},1$) needs only $ 0.9$ seconds (instead of
$ 1.62$ in Table \ref{tb:res}) and for LCG($2^{59},13^{13},0$) it becomes $
0.82$ seconds (instead of $ 1.75$).  All times were measured on a multi-core 64
bit i$7$-processor with $ 2.2$ GHz under the Ubuntu operating system.

While the ratio transformation improved many results
justifying the additional effort of the second random number (and the division
operation), this is not the case with the competitor, the direct-2 approach 
in \eqref{eq:DefStd2}. In many cases results are not improved by an application
of the direct-2 approach (see column `Dir2' in %\Wuu{
Table %}{table} 
\ref{tb:res}). Only for 4
%\Wuu{
RNGs %}{RNGS} 
the %\Wuu{
direct-2 approach %}{Std2-transformation} 
passed all tests in
%\Wu{
both
%} 
\textsc{Crush} and
\textsc{BigCrush}, %\Wu{, } 
and for another one all tests in \textsc{Crush} were passed. In
17 cases, the results became even worse after an application of the direct-2
approach. This may in part be due to the structure of the tests that mainly rely
on the first 32 Bits of the numbers produced and these do not change when
direct-2 is applied.

\section{Summary and conclusion}
\label{sec:Conclu}

In this paper we investigated the impact of a simple ratio
  transformation  on the quality of RNGs, in particular in comparison
  with the direct approach that is used by most congruential generators. The theoretical
  properties showed a less regular, but somewhat  coarser behaviour than the
  direct approach. The statistical tests, however, demonstrated the strength of the ratio:
  it breaks up the linear regularities of its base RNG and turns many mediocre RNGs
  into excellent ones that pass all tests  of \textsc{BigCrush}.

% To measure the impact of this ratio transformation when applied to real
% generators, we carried out an empirical study with a well-known test suite from
% \cite{l2007testu01}. It showed in particular that simple linear congruential
% generators, that fail many of the direct tests are very much improved when the
% output is piped through a ratio transformation.

  Including the ratio transformation in random number generators would increase
  their running time but would also give much better results in many cases.

\section{Appendix: Maximal KS-Distance of the Ratio Transformation}
To complete the proof of Theorem \ref{theo:KSDistance} b), namely $ \Delta_Y\le
\max\{\epsilon_0,\epsilon_1\}$, it remains to show that the maximal distance
apart from the jumps at $ \epsilon_0,1-\epsilon_1$ is also bounded by
$\max\{\epsilon_0,\epsilon_1\} $.

We number the elements of the set $ \D$ as defined in \eqref{eq:DefD}  as $
\D = \{t_1,t_2, $ $\ldots, t_{L-1}\}$ where 
\[ 
t_1=\frac{1}{M-1}<t_2<\cdots<t_{L-1}=\frac{M-2}{M-1}
\]
and set $ t_0:=\epsilon_0, t_L:=1-\epsilon_1$  and $ t_{-1}:=0$. Then we know from \eqref{eq:KSDistance} that 
\[ 
\Delta_Y=\sup_{t\in [0,1]}|F_Y(t)-t|=\max_{i=0 ,\ldots, L}\Big\{|F_Y(t_i)-t_i|,|F_Y(t_{i-1})-t_i|\Big\}
\]

\begin{Theorem}\label{theo:KSAppend}

  \begin{enumerate}
  \item 
%\[ \max_{i=2 ,\ldots,
%      L-1}\Big\{|F_Y(t_i)-t_i|,|F_Y(t_{i-1})-t_i|\Big\}\le
%    \frac{7/4-\sqrt{2}}{M}+\frac{2}{M^2} \le\max\{\epsilon_0,\epsilon_1\}
%\]
\begin{equation}\label{eq:AppendixA}
\max_{i=2 ,\ldots,
      L-1}\Big\{|F_Y(t_i)-t_i|,|F_Y(t_{i-1})-t_i|\Big\}\le
    \frac{7/4-\sqrt{2}}{M}+\frac{2}{M^2} \le\max\{\epsilon_0,\epsilon_1\}
\end{equation}
\item 
%$ \displaystyle \max\Big\{|F_Y(t_1)-t_1|,|F_Y(t_{0})-t_1|\Big\}=
%  |F_Y(t_1)-t_1|\le \epsilon_0$ 
\begin{equation}
 \displaystyle \max\Big\{|F_Y(t_1)-t_1|,|F_Y(t_{0})-t_1|\Big\}=
  |F_Y(t_1)-t_1|\le \epsilon_0
\end{equation}
  \end{enumerate}  
\end{Theorem}
%%%%%%%%%%%%%%%%%%%%%%%%%%%%%%%%%%%%%%%%%%%%%%%%%%%%%%%%%%%%%%%%%%%%%%%%%%%%
\begin{proof} 
    Part a) 
  \begin{enumerate}
    \renewcommand{\theenumi}{\textbf{\arabic{enumi}}}
    \renewcommand{\labelenumi}{\theenumi.}
  \item 
From Theorem \ref{Theo:cdfRatioDiscr} we have
  \begin{align*}
  F_Y(t)-t&=  2\epsilon_0 +\frac{2}{M^2} \sum_{k=1}^{M-1}\fl{tk}-t= 2\epsilon_0
  -\frac{2}{M^2} \Big(\frac{M^2}{2}t-\sum_{k=1}^{M-1}\fl{tk}\Big).
\end{align*}
For $ 1\le
i,j\le L-1, \text{ i.e. for }t_i,t_j\in \D$ define
\begin{align*}
  \Delta(i,j)&:= \frac{M^2}{2}t_i-\sum_{k=1}^{M-1}\fl{t_jk} =\sum_{k=1}^{M-1}
  \big(kt_i - \fl{t_jk}\big)+\frac{M}{2}t_i.
\end{align*}
Then
  \begin{align}
\max_{i=2 ,\ldots, L-1}&\Big\{|F_Y(t_i)-t_i|,|F_Y(t_{i-1})-t_i|\Big\} \label{eq:2}\\
&= \max_{i=2 ,\ldots, L-1}\Big\{|2\epsilon_0 - \frac{2}{M^2}\Delta(i,i)|,|2\epsilon_0 -\frac{2}{M^2}\Delta(i,i-1)|\Big\} \notag\\
&= \max\Big\{ 2\epsilon_0-\frac{2}{M^2}\min_{i=2 ,\ldots, L-1} \{\Delta(i,i),\Delta(i,i-1)\} , \notag\\
& \hspace*{ 3cm}  \frac{2}{M^2}\max_{i=2 ,\ldots, L-1} \{\Delta(i,i),\Delta(i,i-1)\} -   2\epsilon_0 \Big\} \notag.
\end{align}
\item Now let $ t_i=\frac{m}{n}$ and $ T=T(n):=\fl{\frac{M-1}{n}}$, then we have
    for any $ 1\le k\le M-1$
    \begin{equation}\label{eq:Floortk}
    \fl{t_{i-1}k}=
\begin{cases}
  \fl{t_ik}-1&\text{ if }k=nj \text{ for some } 1\le j\le T\\
  \fl{t_ik}&\text{ if }k\not=nj \text{ for all } 1\le j\le T
\end{cases}
\end{equation}
To prove \eqref{eq:Floortk}, note first that, as  $m,n$ are coprime, $ \fl{t_{i}k} =t_{i}k$  holds iff  $ k=nj$
for some $ j\in \N$. If this is the case then $\fl{t_{i-1}k}\le t_{i-1}k<t_{i}k=\fl{t_{i}k}$ and
therefore $\fl{t_{i-1}k}\le\fl{t_{i}k}-1 $. On the other hand 
\[
t_{i}>t_{i}-\frac{1}{k} =\frac{m}{n}-\frac{1}{nj}=\frac{mj-1}{nj} \in \D\cup\{0\},
\] hence $t_{i}-\frac{1}{k}\le t_{i-1}$  and
\[ 
\fl{t_{i-1}k}\ge \fl{(t_{i}-\frac{1}{k})k}=\fl{t_{i}k}-1.
\]
For the second case of \eqref{eq:Floortk}, we have $ \fl{t_{i}k}<t_{i}k$ as $ k\not= nj$. Assume $
t_{i-1}k<\fl{t_{i}k}$, then
\[ 
\fl{t_{i-1}k} \le t_{i-1}k <\fl{t_{i}k}< t_{i}k \quad \text{ and }\quad t_{i-1} <\frac{\fl{t_{i}k}}{k}< t_{i}.
\]
This contradicts the fact that $ t_{i-1}$
and $ t_{i}$ are consecutive elements in $ \D$. Hence \eqref{eq:Floortk} is proved.
\item 
Using \eqref{eq:Floortk}, we have
\begin{align*}
  \Delta(i,i-1)&= \sum_{k=1}^{M-1} \big(kt_i - \fl{t_{i-1}k}\big)+\frac{M}{2}t_i\\
  &= \sum_{k=1}^{M-1} \big(kt_i - \fl{t_{i}k}\big) +T(n)+\frac{M}{2}t_i
  = \Delta(i,i) + T(n)
\end{align*}
as there are exactly $ T(n)=\fl{\frac{M-1}{n}}$ values of $ 1\le k\le M-1$ that have $
k=nj$ for some $j \in \N$.
\item As $ m,n$ in $ t_i=\frac{m}{n}$ are coprime, we have for the central part
  in $ \Delta(i,i)$ 
  \begin{align*}
    \sum_{k=1}^{M-1} \big(kt_i - \fl{t_{i}k}\big) &= \sum_{k=1}^{M-1}
    \big(\frac{km}{n} - \fl{\frac{km}{n}}\big)= \frac{1}{n}\sum_{k=1}^{M-1}\big( mk \mod n\big)\\
&=  \frac{1}{n}\Big(\sum_{\nu=0}^{T-1}\sum_{l=0}^{n-1}\big( (\nu n+l)m
  \mod n\big)+ \sum_{l=0}^d \big( (T n+l)m
  \mod n\big)\Big)\\
%   &= \frac{1}{n}\Big(\sum_{\nu=0}^{T-1}\sum_{l=0}^{n-1}( lm
%   \mod n)+ \sum_{l=0}^d ( lm
%   \mod n)\Big) \\
% &= \frac{1}{n}\Big(T \sum_{l=0}^{n-1}( lm
%   \mod n)+ \sum_{l=0}^d ( lm
%   \mod n)\Big)\\
% &= \frac{1}{n}\Big(T \sum_{l=0}^{n-1}l+ \sum_{l=0}^d ( lm
%   \mod n)\Big)\\
&= \frac{1}{n}\Big(T \frac{n(n-1)}{2}+ \sum_{l=0}^d ( lm
  \mod n)\Big)
  \end{align*}
where $ T=\fl{\frac{M-1}{n}}$ as before and $ d=(M-1)\mod n$. Here, the last sum
may be bounded as follows:
\[ 
\sum_{l=1}^{d}l \ \le\  \sum_{l=0}^d ( lm \mod n) \ \le \ \sum_{l=n-d}^{n-1}l.
\]
\item 
Thus we obtain a lower bound for $ \Delta(i,i)$ as  
  \begin{align*}
    \Delta(i,i)& \ge  
 \frac{1}{n}\Big(T \frac{n(n-1)}{2}+   \sum_{l=0}^dl \Big)+ \frac{Mm}{2n}\ge  \frac{1}{2}\Big(n(T^2 + T) + \frac{M^2}{n}\Big)  -MT \\
&\ge  M(\sqrt{T^2 + T}-T)\ge   M(\sqrt{2}-1)  =: \Delta_{lower}
  \end{align*}
where we used $2 \sqrt{xy} \le x + y$ for $x,y \ge 0$ and $T\ge 1 $.
\item In a similar way we obtain
  \begin{align*}
    \Delta(i,i-1)&=\Delta(i,i)+T\le -\frac{1}{2} \Big(n(T^2+T)
    +\frac{M^2}{n}\Big)   + 3/2M+TM-1\\
    &\le M\Big( 3/2 +T - \sqrt{T^2+T}\Big)-1\\
&\le  M(  5/2  -\sqrt{2})  =:\Delta_{upper}
  \end{align*}
\item Inserting these last two bounds into \eqref{eq:2} we obtain
\begin{align*}
\max_{i=2 ,\ldots, L-1}&\Big\{|F_Y(t_i)-t_i|,|F_Y(t_{i-1})-t_i|\Big\} \\
&    \le
    \max\big\{ 2\epsilon_0-\frac{2}{M^2}\Delta_{lower} ,  \frac{2}{M^2}\Delta_{upper}-2\epsilon_0 \big \}\\
    &< 2  \frac{(7/4-\sqrt{2})M+1}{M^2}
    \le \frac{7/2-2\sqrt{2}}{M}+\frac{2}{M^2} \\
    &= 0.6715729 \frac{1}{M}+\frac{2}{M^2}
\end{align*}
which proves part a) of the Theorem. 
\end{enumerate}
Part b) follows by evaluating $ F_Y(t_1)$ with $ t_1=\frac{1}{M-1}$.
\end{proof}

%\appendixhead{ZHOU}

% Acknowledgments
% \begin{acks}
% The authors would like to thank the anonymous for their helpful
% remarks and suggestions.
% \end{acks}



\section*{Reference}
\bibliographystyle{elsarticle-num} 
\bibliography{rngbib.bib}


%\begin{thebibliography}{00}


%\bibitem{}

%\end{thebibliography}
\end{document}